\documentclass[a4paper,11pt]{amsart}

\usepackage{amsmath, amsthm, amssymb, marvosym, hyperref}

\usepackage{color}

\usepackage{fullpage}

\newcommand{\caA}{{\mathcal A}}
\newcommand{\caB}{{\mathcal B}}

\newcommand{\caH}{{\mathcal H}}

\newcommand{\bbN}{{\mathbb N}}

\newcommand{\ie}{{\it i.e.\/} }

\newcommand{\iu}{\mathrm{i}}
\newcommand{\str}{^*}

\newcommand{\ep}[1]{\mathrm{e}^{#1}}
\newcommand{\dif}{\mathrm{d}}
\newcommand{\Idif}{\,\mathrm{d}}

 \newtheorem{thm}{Theorem}

\begin{document}

\title{Local factorisation for the dynamics of quantum spin systems} 

\author[S. Bachmann]{Sven Bachmann}
\address{Mathematisches Institut der Universit{\"a}t M{\"u}nchen \\ Theresienstr. 39, 80333 M{\"u}nchen \\ Germany}
\email{sven.bachmann@math.lmu.de}

\author[A. Bluhm]{Andreas Bluhm}
\address{Mathematisches Institut der Universit{\"a}t M{\"u}nchen \\ Theresienstr. 39, 80333 M{\"u}nchen \\ Germany}
\email{andreas.bluhm@campus.lmu.de}

\begin{abstract}
Motivated by the study of area laws for the entanglement entropy of gapped ground states of quantum spin systems and their stability, we prove that the unitary cocycle generated by a local time-dependent Hamiltonian can be approximated, for any finite set $X$, by a tensor product of the corresponding unitaries in $X$ and its complement, multiplied by a dynamics strictly supported in the neighbourhood of the surface $\partial X$. The error decays almost exponentially in the size of the neighbourhood and grows with the square of the area~$\vert \partial X\vert^2$.
\end{abstract}

\maketitle

\date{\today }

\section{Introduction}

The general locality result for the dynamics of quantum lattice systems proved in the seminal work of Lieb and Robinson~\cite{Lieb:1972aa} has been refined and extended in many directions, and it has found countless applications. The Lieb-Robinson bound is best understood as expressing the linear growth of the support of local observables, and the Lieb-Robinson velocity is therefore referred to as the `speed of sound' of the quantum lattice system. In this letter, we use it to prove a strong version of locality, namely directly at the level of the unitary cocycle implementing the dynamics in any finite volume, the result being uniform in the volume. The Hamiltonian generating the dynamics can be time-dependent. Indeed, the result applies not only to the actual time evolution of the system, but, in fact more importantly, to the spectral flow studied in~\cite{Bachmann:2012aa}, and used in~\cite{Michalakis:2012wq}. The latter article motivates the present work, which provides a simple proof of Lemma~1 therein.

Concretely, let $\Lambda$ be a finite subset of the lattice and $X\subset\Lambda$. Let $U_\Lambda^{t,s}$ be the unitary evolution generated by the family $H_\Lambda(t)$. We prove that $U_\Lambda^{t,s}$ can be approximated in norm by the product of two unitaries: firstly the purely factorised dynamics $U_X^{t,s}\otimes U_{\Lambda\setminus X}^{t,s}$, and secondly a `surface operator' $\tilde U_{\partial_{R} X}(t,s)\str$ supported on the $R$-fattening of the surface $\partial X$ of $X$. The error decays faster than any inverse power in $R$ and grows like $\vert \partial X \vert^2$. Beyond its structural interest, the result and its particular dependence on the surface of $X$ plays an important role in the understanding of the stability of area laws for the entanglement entropy, see again~\cite{Michalakis:2012wq} as well as~\cite{Marien:2014ub}. The factorisation presented here, which holds in any dimension, should also be put in parallel with the similar factorisation of the ground state projections proved in~\cite{Hastings:2007bu, Hamza:2009di}. We note in particular that the error there grows similarly with the square of the surface of $X$.


\section{Setting}
We consider a quantum spin system defined on a countable set $\Gamma$, which is equipped with its graph distance $d(x,y)\geq 0$ for any $x,y\in \Gamma$. We assume that there is a $F:[0,\infty)\to(0,\infty)$, satisfying the two conditions
\begin{align}
\Vert F\Vert &:= \sup_{x\in\Gamma}\sum_{y\in\Gamma}F(d(x,y))<\infty \label{F_int},\\
C_F&:= \sup_{x,z\in\Gamma}\sum_{y\in\Gamma}\frac{F(d(x,y))F(d(y,z))}{F(d(x,z))} <\infty. \label{F_Conv}
\end{align}
Let $\xi:[0,\infty)\to(0,\infty)$ be a non-increasing function that is logarithmically superadditive, namely
\begin{equation}\label{supadditivity}
\xi(a+b) \geq \xi(a)\xi(b).
\end{equation}
Then the function $F_\xi(r) := F(r)\xi(r)$ satisfies~(\ref{F_Conv}) again, as
\begin{equation*}
\frac{\xi(d(x,y))\xi(d(y,z))}{\xi(d(x,z))} \leq \frac{\xi(d(x,y)\xi(d(y,z)}{\xi(d(x,y) + d(y,z))} \leq 1,
\end{equation*}
for a constant $C_\xi\leq C_F$. We assume moreover that
\begin{equation*}
\lim_{r\to\infty}\xi(r) r^n\to 0
\end{equation*}
for all $n\in\bbN$, in which case $F_\xi$ satisfies~(\ref{F_int}). Furthermore, the function
\begin{equation*}
\zeta(R) := \sum_{r\geq R+1}\xi(r),\qquad R\in\bbN\cup\{0\},
\end{equation*}
is strictly decreasing and decays faster than any inverse power. Indeed, for any $n\in\bbN$,
\begin{equation*}
\lim_{R\to\infty}\zeta(R) R^n\leq \lim_{R\to\infty}\Big(\sup_{r\geq R+1}\xi(r) r^{n+2} \Big)\sum_{r\geq R+1}\frac{(r-1)^n}{r^{n+2}} = 0.
\end{equation*}

To any finite $\Lambda\subset\Gamma$, we associate the finite dimensional Hilbert space $\caH_\Lambda = \otimes_{x\in\Lambda} \caH_x$, where $\mathrm{dim}\caH_x < \infty$. We denote $\caA_\Lambda = \caB(\caH_{\Lambda})$ the set of observables supported in $\Lambda$. It is canonically embedded in $\caA_{\Lambda'}, \Lambda'\supset\Lambda$ by tensoring with the identity.

Let $T\in(0,\infty)$ and $-T\leq s, t\leq T$. The Hamiltonian $H_\Lambda(t)\in\caA_\Lambda$ is a sum of time-dependent local interactions,
\begin{equation*}
H_\Lambda(t) = \sum_{Z\subset\Lambda}\Psi(Z,t),
\end{equation*}
with $\Psi(Z,t) = \Psi(Z,t)\str \in\caA_Z$. We assume that the map $[-T,T]\ni t\mapsto \Psi(Z,t)$ is continuous for all finite $Z$. Furthermore, we assume that
\begin{equation*}
\Vert \Psi \Vert_\xi:= \sup_{x,y\in\Gamma} \frac{1}{F_\xi(d(x,y))} \sum_{Z\ni x,y}\sup_{t\in[-T,T]} \Vert \Psi(Z,t) \Vert<\infty.
\end{equation*}
The self-adjoint $H_\Lambda(t)$ generates the unitary cocycle $U_\Lambda(t,s)\in\caA_\Lambda$ solution of
\begin{equation*}
\iu \frac{\dif}{\dif t}U_\Lambda(t,s) = H_\Lambda (t) U_\Lambda(t,s), \qquad U_\Lambda(s,s) = 1.
\end{equation*}
The associated Heisenberg evolution
\begin{equation*}
\tau_\Lambda^{t,s}(A):= U_\Lambda(t,s)\str A U_\Lambda(t,s),
\end{equation*}
for any $A\in\caA_{\Lambda}$, is such that $\frac{\dif}{\dif t}\tau_\Lambda^{t,s}(A) = \tau_\Lambda^{t,s}(\iu [H_\Lambda (t) , A])$ and that
\begin{equation}\label{cocylce}
\tau_\Lambda^{t,s}\circ\tau_\Lambda^{s,r}(A) = \tau_\Lambda^{t,r}(A).
\end{equation}
The dynamics satisfies a Lieb-Robinson bound of the form
\begin{equation}\label{LRBound}
\left\Vert \left[ \tau_\Lambda^{t,s}(A), B\right] \right\Vert\leq \frac{2\Vert A\Vert \Vert B\Vert}{C_\xi} \min\big[1, g_\xi(t-s)\sum_{x\in X}\sum_{y\in Y}F_\xi(d(x,y))\big]
\end{equation}
where $A \in \caA_X$, $B \in \caA_Y$ and 
\begin{equation*}
g_\xi(r) = \begin{cases}
\exp(v_\xi \vert {r}\vert ) - 1 & \text{if }d(X,Y)>0 \\ \exp(v_\xi \vert {r}\vert ) &\text{otherwise}
\end{cases}
\end{equation*}
see~\cite{Nachtergaele:2006bh}, and~\cite{Bachmann:2012aa} for the time-dependent case. The positive constant $v_\xi = 2\Vert\Psi\Vert_\xi C_\xi $ is the Lieb-Robinson velocity.

Let $X\subset\Gamma$ be a finite subset of $\Gamma$. The inner boundary of $X$ is defined by
\begin{equation*}
\partial X :=  \{x\in X: d(x,\Gamma\setminus X)= 1\}
\end{equation*}
For any $R\in\bbN\cup\{0\}$, consider the sets
\begin{align*}
\overline{X}^R &:= \{x\in\Gamma: d(x,X)\leq R\} \\
\underline{X}_R &:= \{x\in\Gamma: d(x,\Gamma\setminus X)\leq R\} \\
\partial_R X &:= \overline{X}^R \cap \underline{X}_R,
\end{align*}
and note that $\partial X = \underline{X}_1\setminus \underline{X}_0$.

We shall finally assume that there is $G>0$ such that for any finite $X\subset\Gamma$ and all $R$,
\begin{equation}\label{fat boundary}
\vert \partial_R X \vert \leq G R^d \vert \partial X\vert,
\end{equation}
where $\vert\cdot\vert$ denotes the volume of the set, \ie its cardinality.

\section{The factorisation}

In this section, we consider $X$ a fixed finite set such that the map
\begin{equation*}
\bbN\ni n\mapsto \vert \underline{X}_{n}\setminus\underline{X}_{n-1} \vert\geq 0
\end{equation*}
is non-increasing. For any $R\in \bbN\cup\{0\}$ and $\Lambda$ any finite set such that $\overline{X}^{R}\subset\Lambda$, let
\begin{equation*}
M_R := \{Y\subset\partial_{\lfloor R/2 \rfloor}X:  Y\cap X\neq\emptyset \text{ and } Y\cap \Lambda\setminus X\neq\emptyset\},
\end{equation*}
and we shall write $R/2$ instead of $\lfloor R/2 \rfloor$ for simplicity in the following. With this, we define a surface energy by
\begin{equation*}
S(R,t) := \sum_{Z\in M_R} \Psi(Z,t).
\end{equation*}
Note that
\begin{align}
\sup_{t\in[-T,T]}\left\Vert S(R,t) \right\Vert &\leq \sum_{x,y\in \partial_{R/2}X}\sum_{Z\in M_R: x,y\in Z}\sup_{t\in[-T,T]}\Vert \Psi(Z,t)\Vert \leq \vert \partial_{R/2} X\vert \Vert F_\xi\Vert \Vert \Psi\Vert_\xi  \label{SNorm} \\
&\leq G\Vert F_\xi\Vert\Vert \Psi\Vert_\xi (R/2)^d \vert \partial X\vert, \nonumber
\end{align}
by~(\ref{fat boundary}). Finally, let $\tilde U_{\partial_{R} X}(t,s)\in \caA_{\partial_{R} X}$ be the differentiable unitary cocycle defined as the solution of
\begin{equation*}
- \iu \frac{\dif}{\dif t} \tilde U_{\partial_{R} X}(t,s) = \tau_{\partial_{R} X}^{t,s}(S(R,t)) \tilde U_{\partial_{R} X}(t,s),\qquad \tilde U_{\partial_{R} X}(s,s) = 1.
\end{equation*}
\begin{thm}\label{THM}
For a finite $X\subset\Gamma$ as above, $R\in\bbN$, and any finite $\Lambda\subset\Gamma$ such that $\overline{X}^{R}\subset\Lambda$,
\begin{multline*}
\left\Vert U_\Lambda(t,s) - \left( U_X(t,s)\otimes U_{\Lambda\setminus X}(t,s) \right) \tilde U_{\partial_{R} X}(t,s)\str \right\Vert \\
\leq \vert\partial X\vert \Vert \Psi \Vert_\xi \Vert F\Vert \vert t-s\vert  \left[ 2\zeta(R/2) 
 +  \vert \partial X\vert \kappa\, (R/2)^{2d}\xi\left(R/2\right) (\ep{v_\xi \vert t-s\vert}-1) \right],
\end{multline*}
where $\kappa= \frac{G^2 \Vert F_\xi\Vert}{C_\xi}$.
\end{thm}
\begin{proof}
By unitarity,
\begin{equation*}
\left\Vert U_\Lambda(t,s) - \left( U_X(t,s)\otimes U_{\Lambda\setminus X}(t,s) \right) \tilde U_{\partial_{R} X}(t,s)\str \right\Vert = \left\Vert V(t,s)\str -  \tilde U_{\partial_{R} X}(t,s)\str \right\Vert,
\end{equation*}
where
\begin{equation*}
 V(t,s) =  U_\Lambda(t,s)\str \left( U_X(t,s)\otimes U_{\Lambda\setminus X}(t,s) \right).
\end{equation*}
We note that $V(s,s) = 1$ and that
\begin{align*}
\iu \frac{\dif}{\dif t} V(t,s) &= -U_\Lambda(t,s)\str H_\Lambda(t)\left( U_X(t,s)\otimes U_{\Lambda\setminus X}(t,s) \right) \\
&\quad + U_\Lambda(t,s)\str (H_X(t) + H_{\Lambda\setminus X}(t))\left( U_X(t,s)\otimes U_{\Lambda\setminus X}(t,s) \right) \\
&= \tau_\Lambda^{t,s}(K(t)) V(t,s)
\end{align*}
where
\begin{equation*}
K(t) = -\sum_{\substack{Z\subset\Lambda:\\ Z\cap X\neq\emptyset \text{ and } Z\cap \Lambda\setminus X\neq\emptyset}} \Psi(Z,t).
\end{equation*}
By the standard estimate
\begin{align*}
\left\Vert V(t,s)\str -  \tilde U_{\partial_{R} X}(t,s)\str \right\Vert &= \left\Vert V(t,s)\str\tilde U_{\partial_{R} X}(t,s) -1 \right\Vert \\
&\leq \int_s^t\left\Vert \frac{\dif}{\dif r} V(r,s)\str \tilde U_{\partial_{R} X}(r,s)\right\Vert \Idif r \\
&\leq \vert t-s\vert \sup_{r\in[t,s]} \left\Vert \tau_\Lambda^{r,s}(K(r)) + \tau_{\partial_{R}X}^{r,s}(S(R,r)) \right\Vert,
\end{align*}
it suffices to control the sum of the generators, uniformly in $r$ and $\Lambda$.

\noindent \emph{Claim 1.} For any $R\in\bbN$:
\begin{equation*}
\sup_{r\in[s,t]} \Vert K(r) + S(R,r)\Vert\leq 2 \vert \partial X\vert  \Vert \Psi \Vert_\xi \Vert F\Vert \zeta (R/2);
\end{equation*}
\noindent \emph{Claim 2.} For any $R\in\bbN$, and for any $A\in\caA_{\partial_{R/2} X}$:
\begin{equation*}
\left\Vert \tau_\Lambda^{r,s}(A) - \tau_{\partial_{R}X}^{r,s}(A) \right\Vert
\leq \frac{G \Vert F\Vert}{C_\xi}  \Vert A\Vert \vert \partial X\vert  \left(R/2\right)^d\xi\left(R/2\right) (\ep{v_\xi \vert r-s\vert}-1).
\end{equation*}

We first prove Claim~1. By definition,
\begin{equation*}
K(r) + S(R,r) = -\sum_{\substack{Z\subset\Lambda: Z\cap\Lambda\setminus \partial_{R/2} X\neq \emptyset, \\ Z\cap X\neq\emptyset, Z\cap \Lambda\setminus X\neq\emptyset}}\Psi(Z,r),
\end{equation*}
the norm of which can be estimated as
\begin{multline*}
\sup_{r\in[s,t]} \Vert K(r) + S(R,r)  \Vert
\leq \sum_{n\in\bbN }\sum_{x\in \underline{X}_{n}\setminus\underline{X}_{n-1}}
\bigg[\sum_{y\in \Lambda\setminus \overline{X}^{R/2}}\sum_{Z\ni x,y}\sup_{r\in[s,t]}  \Vert\Psi(Z,r) \Vert \\   + \chi(n\geq R/2+1)\sum_{y\in \overline{X}^{R/2}\setminus X}\sum_{Z\ni x,y}\sup_{r\in[s,t]} \Vert\Psi(Z,r) \Vert\bigg].
\end{multline*}
Multiplying and dividing by $F_\xi(d(x,z))$, we indeed obtain
\begin{align*}
\sup_{r\in[s,t]} \Vert K(r) + S(R,r)  \Vert&\leq \Vert \Psi \Vert_\xi \Vert F\Vert \sum_{n\in\bbN }\vert \underline{X}_{n}\setminus\underline{X}_{n-1} \vert \left(\xi(n+R/2) + \chi(n\geq R/2+1)\xi(n) \right)  \\
&\leq 2 \vert \partial X\vert  \Vert \Psi \Vert_\xi \Vert F\Vert \zeta (R/2).
\end{align*}
In the last inequality, we used that, by assumption on $X$,
\begin{equation*}
\vert \underline{X}_{n}\setminus\underline{X}_{n-1} \vert \leq \vert \underline{X}_{1}\setminus\underline{X}_{0}\vert = \vert \partial X\vert.
\end{equation*}

We now prove Claim~2, which is a well-known consequence of the Lieb-Robinson bound in the case of a time independent Hamiltonian. Let $A\in\caA_Y$ with $Y\subset\Lambda'\subset\Lambda$. First, we note that
\begin{equation*}
\frac{\dif}{\dif s}\tau_\Lambda^{t,s}(A) = -\iu \left[ H_\Lambda(s), \tau_\Lambda^{t,s}(A)\right].
\end{equation*}
This follows immediately from the cocycle property~(\ref{cocylce}),
\begin{equation*}
0 = \frac{\dif}{\dif s} \tau_\Lambda^{s,t} \circ \tau_\Lambda^{t,s}(A) = \tau_\Lambda^{s,t}\left( \iu[H_\Lambda(s),\tau_\Lambda^{t,s}(A)]\right) + \tau_\Lambda^{s,t} \left(\frac{\dif}{\dif s} \tau_\Lambda^{t,s}(A)\right),
\end{equation*}
and the fact that $\tau_\Lambda^{s,t}$ is an automorphism. With this,
\begin{equation*}
\frac{\dif}{\dif u}\tau_\Lambda^{u,s}\circ\tau_{\Lambda'}^{r,u}(A) = \tau_\Lambda^{u,s}\left(\iu[H_\Lambda(u) - H_{\Lambda'}(u),\tau_{\Lambda'}^{r,u}(A)]\right),
\end{equation*}
and since $\tau_\Lambda^{r,s}(A) - \tau_{\Lambda'}^{r,s}(A) = \int_s^r \frac{\dif}{\dif u}\tau_\Lambda^{u,s}\circ\tau_{\Lambda'}^{r,u}(A) \Idif r$, we obtain
\begin{align*}
\left\Vert \tau_\Lambda^{r,s}(A) - \tau_{\Lambda'}^{r,s}(A)\right\Vert 
&\leq \sum_{z\in\Lambda\setminus\Lambda'} \sum_{Z\ni z}\int_s^r\left\Vert \left [\Psi(Z,u),\tau_{\Lambda'}^{r,u}(A)\right] \right\Vert \Idif u\\
&\leq \frac{2\Vert A\Vert}{C_\xi}\frac{1}{v_\xi}(\ep{v_\xi \vert r-s\vert}-1) \sum_{z\in\Lambda\setminus\Lambda'} \sum_{Z\ni z} \sup_{u\in[s,r]}\Vert \Psi(Z,u)\Vert \sum_{x\in Z}\sum_{y\in Y} F_\xi(d(x,y))
\end{align*}
by the Lieb-Robinson bound~(\ref{LRBound}). In order to get an upper bound, we replace $\sum_{Z\ni z}\sum_{x\in Z}$ by $\sum_{x\in \Lambda}\sum_{Z\ni x,z}$, multiply and divide by $F_\xi(d(x,z))$ and use the convolution condition to carry out the sum over $x$ to finally get
\begin{align*}
\left\Vert \tau_\Lambda^{r,s}(A) - \tau_{\Lambda'}^{r,s}(A)\right\Vert 
&\leq \frac{\Vert A\Vert}{C_\xi} (\ep{v_\xi \vert r-s\vert}-1)  \sum_{z\in\Lambda\setminus\Lambda'}  \sum_{y\in Y} F_\xi(d(y,z)) \\
&\leq \frac{\Vert A\Vert \Vert F\Vert}{C_\xi} \vert Y\vert \, \xi(d(Y,\Lambda\setminus\Lambda')) (\ep{v_\xi \vert r-s\vert}-1).
\end{align*}
Claim~2 follows by setting $Y = \partial_{R/2}X$ and $\Lambda' = \partial_R X$, so that
\begin{equation*}
\vert Y\vert \,\xi(d(Y,\Lambda\setminus\Lambda')) \leq \vert \partial X\vert G (R/2)^{d}\xi(R/2)
\end{equation*}
by~(\ref{fat boundary}).

Finally, since
\begin{equation*}
\Vert \tau_\Lambda^{r,s}(K(r)) + \tau_{\partial_{R}X}^{r,s}(S(R,r))\Vert
 \leq \Vert \tau_\Lambda^{r,s}(K(r) + S(R,r)) \Vert
 + \Vert \tau_{\partial_{R}X}^{r,s}(S(R,r)) - \tau_\Lambda^{r,s}(S(R,r))\Vert ,
\end{equation*}
Claim~1 and Claim~2 jointly with~(\ref{SNorm}) yield the statement.
\end{proof}
One may be interested in a similar factorisation with the pure tensor product on the right of a surface unitary, namely of the form $\hat U_{\partial_{R} X}(t,s)\str \left( U_X(t,s)\otimes U_{\Lambda\setminus X}(t,s) \right)$. The analog of $V(t,s)$ being here $W(t,s)=\left( U_X(t,s)\otimes U_{\Lambda\setminus X}(t,s) \right)U_\Lambda(t,s)\str$, the necessary cancellations of interaction terms occurs only if one considers the adjoint evolution, namely taking the derivative with respect to $s$. We have
\begin{equation*}
-\iu \frac{\dif}{\dif s} W(t,s) = W(t,s)\sigma_{\Lambda}^{t,s}(K(s)),
\end{equation*}
where $\sigma_{\Lambda}^{t,s}(A) = U_\Lambda(t,s) A U_\Lambda(t,s)\str$. In that case, let $\hat U_{\partial_{R} X}(t,s)\str$ be the solution of
\begin{equation*}
-\iu \frac{\dif}{\dif s} \hat U_{\partial_{R} X}(t,s)\str = \sigma_{\partial_{R} X}^{t,s}(S(R,s)) \hat U_{\partial_{R} X}(t,s)\str,\qquad \hat U_{\partial_{R} X}(t,t)\str = 1,
\end{equation*}
so that
\begin{align}
\label{left}
\left\Vert U_\Lambda(t,s) - \hat U_{\partial_{R} X}(t,s)\str \left( U_X(t,s)\otimes U_{\Lambda\setminus X}(t,s) \right) \right\Vert 
&= \left\Vert W(t,s)\str - \hat U_{\partial_{R} X}(t,s)\str \right\Vert \\
&\leq \int_s^t \left\Vert \sigma_\Lambda^{t,r}(K(r)) + \sigma_{\partial_R X}^{t,r}(S(R,r)) \right\Vert \Idif r.\nonumber
\end{align}
From there, the proof runs along similar lines as above. Claim~1 remains unchanged.  As for Claim~2, we note that $\sigma_\Lambda^{r,s}(A) - \sigma_{\Lambda'}^{r,s}(A) = -\int_s^r\frac{\dif}{\dif u} \sigma_\Lambda^{r,u} \circ \sigma_{\Lambda'}^{u,s}(A)\Idif u$ so that
\begin{equation*}
\left \Vert \sigma_\Lambda^{r,s}(A) - \sigma_{\Lambda'}^{r,s}(A)\right\Vert \leq \sum_{z\in\Lambda\setminus\Lambda'} \sum_{Z\ni z}\int_s^r\left\Vert \left [\Psi(Z,u),\sigma_{\Lambda'}^{u,s}(A)\right] \right\Vert \Idif u.
\end{equation*}
It remains to note that $\Vert [\Psi(Z,u), \sigma_{\Lambda'}^{u,s}(A)]\Vert = \Vert [\tau_{\Lambda'}^{u,s}(\Psi(Z,u)), A]\Vert$ and observe that the Lieb-Robinson bound~(\ref{LRBound}) is symmetric in its arguments to conclude that the upper bound of Theorem~\ref{THM} holds for (\ref{left}) as well.

\end{document}